\documentclass[journal,letterpaper,twocolumn,twoside,nofonttune]{IEEEtran}

\usepackage[T1]{fontenc}
\usepackage{amsmath,amssymb,amsfonts}
\usepackage{mathrsfs}
\usepackage{mathabx}
\usepackage{amsbsy}
\usepackage{graphicx}
\usepackage{graphics}
\usepackage{epstopdf}
\usepackage{algorithm}
\usepackage{algorithmic}
\usepackage{subfigure}
\usepackage{cite}
\usepackage{multirow}
\usepackage{ctable}
\usepackage{caption}

\usepackage{setspace}
\usepackage{mathtools}

\title{\Huge$\,$\\[-2.75ex]
{Secret Key Generation via Pulse-Coupled Synchronization}\\[0.50ex]}
\author{\large%
Hessam Mahdavifar,\,\,\IEEEmembership{Member,~IEEE}, and
Najme Ebrahimi,\,\,\IEEEmembership{Member,~IEEE}\\

\thanks{H.\ Mahdavifar and N.\ Ebrahimi are with the Department of Electrical Engineering and Computer Science, University of Michigan, Ann Arbor, MI 48109 (email: hessam@umich.edu and najme@umich.edu).}
}

\newtheorem{theorem}{{Theorem}}
\newtheorem{lemma}[theorem]{{Lemma}}
\newtheorem{proposition}[theorem]{{Proposition}}
\newtheorem{corollary}[theorem]{{Corollary}}

\newcommand{\cS}{{\cal S}}

\DeclareMathAlphabet{\mathbfsl}{OT1}{ppl}{b}{it} 

\newcommand*{\rom}[1]{\expandafter\romannumeral #1}

\makeatletter
\newcommand{\AlignFootnote}[1]{%
  \ifmeasuring@
  \else
    \iffirstchoice@
      \footnote{#1}%
    \fi
  \fi}
\makeatother

\newcommand{\be}[1]{\begin{equation}\label{#1}}
\newcommand{\ee}{\end{equation}} 
\newcommand{\eq}[1]{(\ref{#1})}

\renewcommand{\leq}{\leqslant}
 
\renewcommand{\geq}{\geqslant}

\newcommand{\script}[1]{{\mathscr #1}}

\newcommand{\Pref}[1]{Pro\-po\-si\-tion\,\ref{#1}}
\newcommand{\Lref}[1]{Lem\-ma\,\ref{#1}}
\newcommand{\Cref}[1]{Co\-ro\-lla\-ry\,\ref{#1}}

\newcommand{\deff}{\mbox{$\stackrel{\rm def}{=}$}}

\newcommand{\sX}{\script{X}}
\newcommand{\sY}{\script{Y}}
\newcommand{\sZ}{\script{Z}}

\newcommand{\shalf}{\mbox{\raisebox{.8mm}{\footnotesize $\scriptstyle 1$}
\footnotesize$\!\!\! / \!\!\!$ \raisebox{-.8mm}{\footnotesize
$\scriptstyle 2$}}}

\newcommand{\ep}{\epsilon}
\newcommand{\mt}{\tilde{m}}

\begin{document}

\maketitle

\begin{abstract}

A novel framework for sharing common randomness and generating secret keys in wireless networks is considered. In particular, a network of users equipped with pulse oscillators (POs) and coupling mechanisms in between is considered. Such mechanisms exist in synchronized biological and natural systems, and have been exploited to provide synchronization in distributed networks. We show that naturally-existing initial random phase differences between the POs in the network can be utilized to provide \textit{almost identical} common randomness to the users. This randomness is extracted from the synchronization time in the network. Bounds on the entropy of such randomness are derived for a two-user system and a conjecture is made for a general $n$-user system. Then, a three-terminal scenario is considered including two legitimate users and a passive eavesdropper, referred to as Eve. Since in a practical setting Eve receives pulses with propagation delays, she can not identify the exact synchronization time. A simplified model is then considered for Eve's receiver and then a bound on the rate of secret key generation is derived. Also, it is shown, under certain conditions, that the proposed protocol is resilient to an active jammer equipped with a similar pulse generation mechanism.

\end{abstract}


\section{Introduction} 
\label{sec:Introduction}

Physical layer security methods provide an alternative to conventional encryption schemes in order to ensure security in wireless networks \cite{bloch2011physical}. Alternatively, they can be deployed to exchange secret keys between the nodes in order to complement the higher layer encryption schemes. The fundamental works of \cite{ahlswede1993common, maurer1993secret} established the use of common randomness for secret key generation. An important question is then how to generate common randomness at the nodes in order to utilize such protocols in wireless networks. To this end, properties of wireless links, such as channel gain and delay are shown to provide a great source for the common randomness, which have recently received significant attention \cite{zhang2016key}. 

There are several challenges, however, to standardize channel-based secret key generation protocols. A common assumption in such protocols is the channel reciprocity between the legitimate parties \cite{cgc,gollakota2011physical}. This would require a perfect synchronization to avoid phase and frequency mismatch between the wireless nodes which is often hard to ensure in distributed networks \cite{simeone2008distributed}. Furthermore, if the nodes are static, then with no channel variations the amount of secret key bits that can be generated will be limited. To resolve this, induced randomness can be introduced in wireless nodes to increase the rate of secret key generation \cite{glob2}. However, in general, such channel-based secret key generation protocols require an extra level of key reconciliation over the public channel to ensure the generated keys match at both ends. This would require an entire standardization of channel coding and modulation for this purpose which would make a barrier in deploying such methods in practice. 

We recently proposed a novel approach to implementation of physical layer security by exploiting coupling dynamics in the network \cite{glob1}. Such coupling dynamics are already being used for synchronization in wireless networks. In particular, we suggested to use coupled oscillators to implement the proposed approach in radio-frequency (RF) front end \cite{glob1}. It is well-known that a network of RF coupled oscillators converges within nanoseconds to a steady-state condition provided that initial free-running frequencies are within a certain \textit{locking range} \cite{ebrahimi2016robustness, ebrahimi201771}. However, such coupling dynamics, such as electromagnetic coupling, are often limited to short distance ranges for this specific application.  

In this paper, we propose to exploit synchronization mechanisms based on pulse-coupled oscillators in order to securely generate random keys in distributed networks. The proposed methods do not require extra processing, e.g., the shared randomness is almost identical at the nodes, and extra hardware to generate randomness. They also do not require channel randomness and consequently, assumptions on channel reciprocity. The naturally-existing random phase differences between the wireless nodes, prior to synchronization, would serve as the source of common randomness. It is shown, under a simplified model for the eavesdropper, that a positive-rate secret key can always be guaranteed. Furthermore, the resilience of the proposed protocol to certain jamming attacks is discussed.

The rest of this paper is organized as follows. In Section\,\ref{sec2} some background on pulse-coupled oscillators and secret key generation is provided. In Section\,\ref{sec3} the system model is formulated. In Section\,\ref{sec4} bounds on the entropy of shared common randomness are derived. In Section\,\ref{sec5} secret key generation rate in the presence of Eve is characterized. Resilience of the proposed protocols to jamming attacks is discussed in Section\,\ref{sec6}. Finally, the paper is concluded in Section\,\ref{sec7}.

\section{Preliminaries} 
\label{sec2}

\subsection{Pulse coupled oscillators}
\label{subsec21}

A pulse oscillator (PO) is characterized using a state variable $x$ which increases monotonically toward a normalized threshold of $x=1$. In the model considered in \cite{str}, $x$ evolves as $x = f(\phi)$, where $\phi$ is the normalized time that increases from $0$ to $1$. Also, $f(0) = 0$ and $f(1) = 1$. The PO transmits a pulse, which can be ideally considered as a delta function with width $0$, once its state $x$ reaches $1$. Then $\phi$ is reset to $0$. The function $f$ is assumed to be concave down and strictly increasing, i.e., $f' > 0$ and $f'' < 0$ over $[0,1]$. Also, we assume that $f$ is semi-diffrentiable at $0$ and $1$, and hence $f'$ is bounded over $[0,1]$. An example of $x$ in an electrical system is the charge of a capacitor in a resistor-capacitor (RC) circuit as a function of time. This matches with the Peskin model \cite{peskin1975mathematical}, which considers $f(\phi) = c(1-e^{-\gamma \phi})$, where $c,\gamma$ are constants. 

Networks of coupled POs naturally exist in synchronized biological and natural systems \cite{peskin1975mathematical,str}, and have been exploited to provide synchronization in distributed wireless networks \cite{simeone2008distributed}. Such a network is modeled as follows. Suppose that each node in the network is equipped with an identical PO. Let $\epsilon \in (0,1)$ be a fixed parameter. When a node $V$ receives a pulse from one of its neighbors, denoted by $U$, in the network, its dynamic changes as follows. If the current state of $V$, $x_V$, is at least $1-\epsilon$, then it is changed to $x_V = 1$ and a pulse is transmitted by $V$ right away. This implies that $U$ and $V$ are synchronized. Otherwise, when $x_V<1-\epsilon$, the state of $V$ is changed to $x_V+\ep$, i.e., its phase $\phi$ is changed to $\phi' = f^{-1}(f(\phi)+\epsilon)$. This can be thought as applying an extra charge of $\epsilon$ to $V$'s capacitor upon arrival of an external pulse. For simplicity, suppose that the network is fully connected, in which case when two nodes become synchronized, they stay synchronized moving forward. It is proved in \cite{str} that, for any $n$, network synchronization occurs in a fully connected network of $n$ identical POs, i.e., all the nodes synchronize to each other, except for a measure-zero set of initial phases $(\phi_1,\phi_2,\dots,\phi_{n})$, where $\phi_i$ is the initial phase of the $i$-th PO. 

\subsection{Secret key generation}
\label{subsec22}

Secret key generation protocols aim at securely establishing random keys between legitimate parties using common randomness. In this paper, we mostly focus on a case involving two legitimate parties Alice and Bob together with a passive Eve. In particular, a three-terminal source-type model is considered. Such model, in general, can be described as follows. Let $X \in \sX$, $Y \in \sY$, and $Z \in \sZ$ denote Alice's, Bob's, and Eve's observations, respectively, where $\sX$, $\sY$, and $\sZ$ are the corresponding alphabets. Following the convention, let capital letters denote the random variables and small letters denote their instances. In the considered source-type model, $X,Y,Z$ are distributed according to a joint probability distribution $p_{X,Y,Z}$. The goal for Alice and Bob is to agree on a shared secret key $K$, based on their observations $X$ and $Y$ using an arbitrary number of communication rounds over a public channel with unlimited capacity. Such process is tightly related to Slepian-Wolf compression. The connection is useful for designing the so-called \textit{key reconciliation} stage of secret key generation protocols using cosets of practical error-correcting codes. The security of $K$ is measured in an information-theoretic sense given Eve's observation $Z$ and all the public interactions between Alice and Bob. The two-user secret key capacity, denoted by $\cS(X;Y|Z)$, is bounded as follows \cite[Theorem 2 and 3]{maurer1993secret}:
\be{key-bound}
\begin{split}
&\max \{I(X; Y ) - I(X;Z), I(Y ; X) - I(Y ;Z)\} \leq S (X; Y|Z)\\
 &\leq  \min \{I(X; Y ), I(X; Y |Z)\}. 
 \end{split}
\ee
Such results were later extended to multiple-terminal scenarios \cite{csiszar2004secrecy}. Furthermore, an exact characterization was derived for a case in which only one round of communication occurs from Alice to Bob \cite[Theorem 2]{ahlswede1993common}.

\section{System Model} 
\label{sec3}

Consider a fully connected network where the network nodes are equipped with identical POs. Suppose that each PO starts with a random phase that is uniformly distributed over $[0,1]$. Then the POs enter a dynamic system as described in Section\,\ref{subsec21}. Each node counts the number of pulses its PO transmits till network synchronization occurs. The network synchronization can be identified by individual nodes once no external pulses are received between two consecutive pulses. Each node saves this number as the common randomness. The following lemma shows that all the nodes observe the same number, up to a difference of $1$. Let $m_i$ denote the number of pulses that the node $V_i$, for $i=1,2,\dots,n$, has transmitted so far at a given time. 

\begin{lemma}
For any $i,j$, we have $|m_i - m_j| \leq 1$. 
\end{lemma}

\begin{proof}
The proof is by noting that before $V_i$ and $V_j$ become synchronized, it is not possible that $V_i$ transmits two consecutive pulses without $V_j$ transmitting any pulse in between. In fact once $V_i$ sends a pulse, we have $\phi_i = 0$, while $\phi_j > 0$. Now, since  $f^{-1}(f(\phi)+\epsilon)$ is a strictly increasing function, we have $\phi_j > \phi_i$ which holds when pulses external to $V_i$ and $V_j$ arrive as well. This holds till $\phi_j = 1$, in which case $V_j$ sends a pulse before $V_i$ sends the next one. 
\end{proof}
The lemma implies that the number of pulses that each PO counts till network synchronization occurs can serve as a source of \emph{almost} noise-free common randomness.

Let us refer to a time-interval during which each of the POs send one pulse as a \textit{full cycle}. Then time is split into non-overlapping and consecutive \textit{full cycles}. Note that since the time between two consecutive pulses by each of the POs keep changing, we can not define a \textit{full cycle} as a fixed time interval. Due to delays in propagation of pulses, an external user/eavesdropper can not exactly identify when synchronization occurs. Under propagation delays, and assuming the delays between POs is less than half the time unit, the synchronization still occurs for $n=2$ \cite{ernst1995synchronization}, and under certain conditions for general $n$ \cite{ernst1995synchronization} and also for locally connected networks \cite{ferrari2017convergence}. To this end, a certain \textit{refractory period} $\rho$ is defined and a PO does not update its state for $\rho$ seconds right after it sends a pulse.  

Motivated by practical considerations of propagation delays we describe Eve's observation as follows. During each \textit{full cycle}, Eve receives $n$ delayed pulses from the $n$ POs. Then she can process the timings between received pulses, compare them with her estimates of propagation delays with each of the nodes, and also compare the timings with previous \textit{full cycles}. Taking all these information into account to model Eve's receiver is not an easy task. Instead, we consider s simplified binary symmetric channel-type model for Eve as follows. Let a binary random variable $S$ indicates whether all nodes are synchronized in the current \textit{full cycle} or not, e.g., if synchronization occurs/has occurred, then $S=1$ and otherwise, $S=0$. Then Eve observes $Z=S$ with probability $1-p$, and $Z=1-S$ with probability $p$, for some $p \in (0,\shalf)$. Also, we consider a memoryless model, in which Eve's observation noise $Z \oplus S$ is independent across different \textit{full cycles}.

Most of prior work on physical layer security involves a \textit{passive} eavesdropper. In our proposed framework, an active eavesdropper may try to act as a legitimate node of the network by deploying similar pulse-coupling mechanisms in order to detect the synchronization time which will be used for key generation. However, such a malicious act can be detected by other nodes assuming they know the total number of nodes in the network. In this paper, we do not formulate such active eavesdropping methods and leave it for future work. There may exist, however, another type of adversary interested in \textit{jamming} the proposed protocol by randomly/selectively sending pulses into the network in order to prevent synchronization. We model this scenario by assuming that the jammer has the same pulse generation mechanism as other nodes in the network, i.e., it can send at most one pulse during each \textit{full cycle}. Also, suppose that legitimate nodes can not distinguish between pulses sent by other legitimate nodes and the jammer. We will show that synchronization may not occur and provide an upper bound on the probability of such event in Section\,\ref{sec6}.

\section{Entropy of Shared Randomness via Pulse-Coupled Synchronization}
\label{sec4}

Let $M$ denote the random variable that represents the total number of pulses before synchronization. In order to simplify the formulation, we take the maximum of the counted pulses by POs as the shared randomness $M$, knowing that each of the POs has counted either $M$ or $M-1$ pulses. At the end of this section, we discuss how such inconsistency can be resolved. 

Let $\left\{p_i\right\}_{i=1}^{\infty}$ denote the probability distribution of $M$, where $p_i = Pr\left\{M = i\right\}$. The goal is to upper bound and lower bound $p_i$'s in order to provide bounds on the entropy of $M$. 

We consider only two nodes. For two oscillators, we show that, roughly speaking, the probability distribution of the number of pulses before synchronization occurs behave like a discretized exponential distribution. 

Similar to \cite{str}, let 
\be{h-def}
h(\tau) \,\deff\, f^{-1}(\ep+f(1-\tau)),\ R(\tau) \,\deff\, h(h(\tau)).
\ee
Let $\delta = 1-f^{-1}(1-\ep)$. Then the domain of $h$ is $(\delta,1)$ and the domain of $R$ is $(\delta,h^{-1}(\delta))$. 

Let $\phi, \phi + \tau$ denote the initial phases of the two POs, where $0<\phi<\phi+\tau<1$. If $\tau \leq \delta$, i.e., $f(1-\tau)\geq 1-\ep$, then the two POs synchronize after the next pulse. Otherwise, the phase difference, after the first pulse is sent, become $h(\tau)$, where $h(.)$ is defined in \eq{h-def}. Hence, $R(\tau)$ is the phase difference after the \textit{full cycle}. Then it is shown in \cite[Proposition 2.2.]{str} that $R$ has a fixed point $\tau^*$ that is a repeller, i.e., for $\tau < \tau^*$, $R(\tau) < \tau$, and for  $\tau > \tau^*$, $R(\tau) > \tau$. Furthermore, it is shown in \cite[Lemma 2.1]{str} that $h' < -1$ and $R' > 1$ over their domains. Since $f'$ is bounded over $[0,1]$, it can be shown that, $\sup |h'| < \infty$, $\inf R' > 1$, and $\sup R' < \infty$ over their domains. Then let 
\be{lamb-def}
\lambda_0 \,\deff\, 1/|\sup h'|,\ \lambda_1 \,\deff\, 1/\sup R',\ \lambda_2 \,\deff\, 1/\inf R',
\ee 
where $0<\lambda_1<\lambda_2<1$, and $\lambda_0 > 1$ by \cite[Lemma 2.1]{str}. Let also $\tau^*$ denote the fixed point of $R$. 

\begin{lemma}
\label{lemfxp}
There exists an increasing sequence of $\left\{\tau_i\right\}_{i=1}^{\infty}$ and a decreasing sequence of $\left\{\tau'_i\right\}_{i=1}^{\infty}$ such that
\begin{itemize}
\item (i) $\lim_{i \rightarrow \infty} \tau_i = \lim_{i \rightarrow \infty} \tau'_i = \tau^*$.
\item (ii) For initial phase difference $\tau \in [\tau_i,\tau_{i+1}] \cup [\tau'_{i+1}, \tau'_i]$, we have $M = i$. 
\end{itemize}
\end{lemma}

\begin{proof}
Let $\tau_0 = 0$, $\tau'_0 = 1$, $\tau_1 = \delta$, $\tau'_1 = h^{-1}(\delta)$, where $\delta = 1-f^{-1}(1-\ep)$. Then for $i\geq 1$, let $\tau_{i+1} = R^{-1}(\tau_i)$ and $\tau'_{i+1} = R^{-1}(\tau'_i)$. To prove the first condition, note that $\tau^*$ is also a fixed point for $R^{-1}$. Also, since $R$ is a repeller, $R^{-1}$ is a contraction mapping. Hence, $(i)$ follows. The proof of the second part is by induction on $i$. Note that if the initial phase diffrence $\tau$ is in $[0,\delta]$, then synchronization occurs after the first pulse is sent. Hence, $M=1$. If $\tau \in [h^{-1}(\delta),1]$, then after the first pulse the phase difference becomes $h(\tau) \in [0,\delta]$. Hence, synchronization occurs after each of the POs send one pulse and again, $M=1$. Now, suppose that $\tau \in [\tau_i,\tau_{i+1}] \cup [\tau'_{i+1}, \tau'_i]$, where $i>0$. After a \textit{full cycle} the phase difference becomes $R(\tau)$ which belongs to $[\tau_{i-1},\tau_{i}] \cup [\tau'_{i}, \tau'_{i-1}]$, and the proof follows by induction hypothesis.
\end{proof}

Let
\be{ab-def} 
a_i = \tau_i - \tau_{i-1},\ b_i = \tau'_{i-1} - \tau'_i,
\ee
for any $i \geq 1$, where $\left\{\tau_i\right\}$ and $\left\{\tau'_i\right\}$ are as introduced in the proof of \Lref{lemfxp}. Then we have the following corollary.

\begin{corollary}
\label{cor1}
Assuming that the initial phase difference is uniform we have $p_i= a_i+b_i$, where $p_i = Pr\left\{M = i\right\} $. 
\end{corollary}

\begin{lemma}
\label{p-bound}
For any $i \geq 2$, $\lambda_1 p_{i-1} \leq p_i \leq \lambda_2 p_{i-1}$, where $\lambda_1,\lambda_2$ are defined in \eq{lamb-def}.
\end{lemma}

\begin{proof}
Let $\tau = \tau_{i}$, for some $i \geq 2$, where $\left\{\tau_i\right\}$ is defined in the proof of \Lref{lemfxp}. Then we have
$$
a_{i} = \tau - R(\tau), a_{i-1} = R(\tau) - R(R(\tau)),
$$ 
where $\left\{a_i\right\}$ is defined in \eq{ab-def}. Since $R$ is a continuous and differentiable function, then by the mean value theorem, there exists $c \in [\tau,R(\tau)]$ such that
$$
R'(c) = \frac{R(R(\tau)) - R(\tau)}{R(\tau) - \tau} = \frac{a_{i-1}}{a_{i}} .
$$
Then by definition of $\lambda_1,\lambda_2$ in \eq{lamb-def} we have 
$$
\lambda_1 \leq \frac{a_{i}}{a_{i-1}} \leq \lambda_2.
$$
The same argument can be applied to $b_i$'s and then the lemma follows by \Cref{cor1}. 
\end{proof}

Note that \Lref{p-bound} implies that the distribution of $M$ resembles a discretized exponential distribution, i.e., $p_i$ decays exponentially fast as $i$ grows. In particular, it is shown in the next proposition that $M$ has a bounded entropy. It is assumed that $p_1< 1/e \approx 0.37$. If $p_1$, and possibly $p_2$, are greater than $1/e$, then we can exclude them, apply the following proposition to the rest of $p_i$'s, and add the terms corresponding to $p_1$ and $p_2$ in $H(M)$ as constants. 

\begin{proposition}
\label{thm-entropy}
In the two-user pulse coupling system, we have
$$
\frac{g(c)}{1-\lambda_1} + \frac{c g(\lambda_1)}{(1-\lambda_1)^2} \leq H(M) \leq \frac{g(c)}{1-\lambda_2} + \frac{c g(\lambda_2)}{(1-\lambda_2)^2},
$$
where $g(x) = -x\log x$, $\lambda_1,\lambda_2$ are defined in \eq{lamb-def}, and $c = p_1 = 1+\delta-h^{-1}(\delta)$. 
\end{proposition}
\begin{proof}
By \Lref{p-bound} we have $p_1 \lambda_1^{i-1} \leq p_i \leq p_1\lambda_2^{i-1}$. Then the proof follows by the definition of entropy function $H(.)$ and noting that $g(x)=-x\log x$ is increasing for $x\leq 1/e$.  
\end{proof}

\noindent
{\bf Remark\,1.} We conjecture that in a general set-up consisting of $n$ POs, $H(M) = O(\log n)$. More specifically, we believe it can be shown that after $O(n)$ \textit{full cycles} the POs are split into $O(1)$ clusters, each consisting of synchronized POs. Then one can only analyze the entropy of shared randomness involving a constant number of POs and use the bound on the entropy of sum of two random variables to prove the conjecture.

\noindent
{\bf Remark\,2.} In order to resolve the possible difference of $1$ between counted pulses by the two users, a simple key reconciliation method can be deployed as follows. Users will exchange their observations modular $3$. Then if there is a difference, the user with smaller observation increments it by $1$ which would make it an error-free common randomness. In general, if we want to recover from a larger difference, up to a certain $d$, between observations, e.g., due to an initial phase difference of more than one unit of time, the same procedure can be deployed. In that case, users exchange their observation modular $2d+1$ using which reconciliation can be done.

\section{Secret Key Rate}
\label{sec5}

A three-terminal model, as described in Section\,\ref{subsec22}, is considered. The common randomness $M$ between Alice and Bob is generated according to the process discussed in Section\,\ref{sec4} with a slight modification as follows. In order to avoid long waiting times, Alice and Bob set a fixed threshold $\mt$. They continue to send pulses until each of them sends $\mt$ pulses, regardless of whether synchronization has occurred or not, at which point they stop the current \textit{session}. Then they may start a new \textit{session} with new initial random phases, e.g., by simply reseting their POs, and the same process will be repeated. If synchronization has occurred at some point during the \textit{session}, then the common randomness $M$ is the number of pulses till that point. Otherwise, $M=\mt$. In other words, the probability distribution $(p_1,p_2,\dots)$, characterized in \Cref{cor1}, is truncated as follows. For $i=1,2,\dots,\mt-1$, $Pr\{M=i\} = p_i$, and $Pr\{M=\mt\} = \sum_{i=\mt}^{\infty} p_i$. Also, for $i>\mt$, $Pr\{M=i\} = 0$. Furthermore, we assume that the common randomness $M$ is identical at Alice and Bob, i.e., $X=Y = M$, and $\sX = \sY = \{1,2,\dots,\mt\}$. 

\noindent
{\bf Remark\,3.} In order to recover from possible differences between Alice's and Bob's observations a procedure, as discussed in Remark\,2, can be deployed. Since $M \mod (2d+1)$ is then revealed to Eve, Alice and Bob take $\left\lfloor\ M/(2d+1) \right\rfloor$ as the common randomness. The bounds provided in this section can be also modified to reflect this extra step, however, we keep assuming $M$ as the common randomness to simplify derivations. 

Eve's observation $Z$, according to the model described in Section\,\ref{subsec21}, is as follows. Let $S=\{S_i\}_{i=1}^{\mt}$ denote the synchronization indicator sequence, where $S_i$ is the indicator of synchronization in the $i$-th \textit{full cycle}, as defined in Section\,\ref{sec3}, for $i=1,2,\dots,\mt$. Note that if $M=m$, then we have $S_i = 0$, for $1 \leq i < m$, and $S_i = 1$, for $m \leq i \leq \mt$. Then $Z = (Z_1,Z_2,\dots,Z_{\mt})$, where $Z_i = S_i \oplus Q_i$, and $Q_i$'s are i.i.d. with Ber$(p)$, where $p$ is the model parameter described in Section\,\ref{sec3}. 

In general, when $X=Y=M$ in the three-terminal model, the lower and upper bounds in \eq{key-bound} match. Hence, the  secret key capacity, which can be denoted by $\cS(M|Z)$, is given as follows:
\be{key-bound2}
\cS(M|Z) = H(M) - I(M;Z) = H(M|Z).
\ee

Since the complexity of the exact computation of $H(M|Z)$ is exponential in terms of $\mt$, we provide a lower bound on $\cS(M|Z)$ in terms of the parameters of the coupling system as well as Eve's parameter $p$. The lower bound shows that the secret key rate $\cS(M|Z)$ is strictly positive regardless of the choice for $\mt$. The following lemma is useful to derive such a bound.

\begin{lemma}
\label{lem-mbound}
For any $m_1,m_2 \in \{1,2,\dots,\mt\}$ we have
$$
\frac{Pr\{Z| M=m_1\}}{Pr\{Z|M=m_2\}} \geq (\frac{p}{1-p})^{|m_1 - m_2|},
$$
for any instance of $Z$. 
\end{lemma}
\begin{proof}
Let $\{s_{i,j}\}_{i=1}^{\mt}$ denote the synchronization indicator sequences for $m_j$, $j=1,2$. Also, note that
$$
Pr\{Z| M=m_j\} = \Pi_{i=1}^{\mt} Pr\{Z_i | S_i = s_{i,j}\}.
$$
This together with noting that $\{s_{i,1}\}$ and $\{s_{i,2}\}$ differ in exactly $|m_1 - m_2|$ positions, and the assumption on the noise $Q_i = Z_i \oplus S_i$ (i.i.d. with Ber$(p)$) complete the proof. 
\end{proof}

\begin{proposition}
\label{prop-main}
For the considered three-terminal model with common shared randomness $M$ the secret key rate $\cS(M|Z)$ is lower bounded as
\be{rate-bound}
\cS(M|Z) \geq \log \min\{ \frac{1}{1-\delta_1}, 1+(1-\lambda_2)\delta_2\frac{1-\delta_2^{\mt}}{1-\delta_2} \},
\ee
where $\delta_1 = \lambda_1p/(1-p)$, $\delta_2 = \lambda_2^{-1} p/(1-p)$, and $\lambda_1,\lambda_2$ are defined in \eq{lamb-def}.
\end{proposition}
\begin{proof}
For any $m_0 \in \{1,2,\dots,\mt\}$, using the Bayes' rule and the law of total probability we have
\be{prop-main-eq1}
Pr\{M=m_0|Z\} = \frac{Pr\{Z|M=m_0\}Pr\{M=m_0\}}{\sum_{m=1}^{\mt} Pr\{Z|M=m\}Pr\{M=m\}}.
\ee
By plugging the bounds from \Lref{lem-mbound} and \Lref{p-bound} in \eq{prop-main-eq1}, for $m_0 \in \{1,2,\dots,\mt-1\}$ we have
\be{prop-main-eq2}
Pr\{M=m_0|Z\} \leq 1/\sum_{i=0}^{\infty} \lambda_1^i  (p/1-p)^i = 1 - \delta_1,
\ee
for any instance of $Z$. And, similarly, for $\mt$ we have
\be{prop-main-eq3}
\begin{split}
Pr\{M=\mt |Z\} &\leq 1/\bigl(1+\sum_{i=1}^{\mt-1} \lambda_2^{-i}  (p/1-p)^i (1-\lambda_2)\bigr)\\
& = 1/\bigl(1+(1-\lambda_2)\delta_2\frac{1-\delta_2^{\mt}}{1-\delta_2}\bigr),
\end{split}
\ee
where we again used bounds from \Lref{lem-mbound} and \Lref{p-bound} in \eq{prop-main-eq1} while noting that $Pr\{M=\mt\} = \sum_{i=\mt}^{\infty} p_i$. 

As stated in \eq{key-bound2}, $\cS(M|Z) = H(M|Z)$. Note that for any two random variables $M,Z$, we have
$$
H(M|Z) \geq -\log \max_{m,z} Pr\{M=m|Z=z\}.
$$
This together with bounds in \eq{prop-main-eq2} and \eq{prop-main-eq3} complete the proof. 
\end{proof}
\noindent
{\bf Remark\,4.} Note that the second term over which minimization of \eq{rate-bound} takes place is increasing with $\mt$. Therefore, the lower bound provided in \Pref{prop-main} is non-decreasing with $\mt$. Hence, the secret key rate, in terms of bits/session, is strictly positive regardless of $\mt$ and is actually bounded away from $0$ as $\mt \rightarrow \infty$. Also, there is a certain threshold such that increasing $\mt$ beyond that threshold does not improve the lower bound of \eq{rate-bound}. It would be interesting to see if the actual secret key rate $\cS(M|Z)$ exhibits the same behavior. 

\noindent
{\bf Remark\,5.} Here, we do not discuss a coding method for Alice and Bob to extract a secure key $K$ from $M$ in an information-theoretic sense, i.e., $I(K;Z)$ being small. In particular, since $H(M)$ is bounded, as shown in Section\,\ref{sec4}, there is no asymptotic behavior for such information-theoretic arguments. One has to consider several \textit{sessions} between Alice and Bob and then apply standard key extraction techniques, e.g., using polar codes \cite{chou2015polar}, to a sequence of shared symbols $M_j$'s. Alternatively, an \textit{ad-hoc} solution is also possible by simply applying a linear transform, e.g. a polarization matrix \cite{arikan2009channel}, to the sequence $\{S_i\}_{i=1}^{\mt}$ in one session. Although $S_i$'s are not independent, an argument similar to \cite[Proposition 3]{mahdavifar2011achieving} can be used to show that compression of $M$ and also extracting a secure $K$ can be done to some extent (Again, no concrete results can be made here as there is no asymptotic behavior). Such solutions are low complex and can be locally and identically performed by Alice and Bob without any need for further communication. The details are left for future work. 

\section{Resilience Against Jamming Attacks}
\label{sec6}

Consider a two-user scenario where legitimate users want to establish synchronization and to extract a common randomness, as discussed in Section\,\ref{sec4}. Suppose that there is a jammer present in the network equipped with a similar PO, as modeled in Section\,\ref{sec3}.

Let $\lambda_0$ be as defined in \eq{lamb-def}. This parameter is frequently used throughout this section. Let $\tau$ denote the phase difference between the two POs at the beginning of a \textit{full cycle}. If there is no jammer, then the phase difference becomes $R(\tau)$ at the end of the \textit{full cycle}, as discussed in Section\,\ref{sec4}. In the presence of a jammer, the dynamic may change as will be described in the following lemma. Note that in the remaining of this section we discard specifying the domains of $h$ and $R$ and assume the following: if $\tau > 1$, then $h(\tau)$ is replaced by $0$; if $\tau > h^{-1}(\delta)$, then $R(\delta)$ is replaced by $1$; if $\tau < \delta$, then $h(\tau)$ is replaced by $1$ and $R(\tau)$ is replaced by $0$, where $\delta$ is defined in Section\,\ref{sec4} as $\delta = 1-f^{-1}(1-\ep)$.

\begin{lemma}
\label{lem-jam}
Let $\tau,\tau'$ denote the phase differences at the beginning and at the end of a \textit{full cycle}. Then
$$
\tau' \in \bigl[h\bigl(\lambda_0 h(\tau)\bigr), \max\{R(\lambda_0 \tau),\lambda_0 R(\tau)\}\bigr],
$$
where $h$, $R$, and $\lambda_0$ are defined in \eq{h-def}, and \eq{lamb-def}, respectively. 
\end{lemma}
\begin{proof}
There are three possible scenarios for the arrival time of jammer's pulse during one \textit{full cycle} in terms of how many pulses, either $0$, $1$, or $2$ pulses, have been sent in that cycle. Consider the first case where jammer's pulse arrives at a time $t$ before any of the two POs send a pulse. Let $\phi, \phi + \tau$ denote the phases of the two POs at time $t$. Then the phases change to  $h(1-\phi)$, $h(1-\phi-\tau)$ right after $t$ (Note that $h$ is a decreasing function with $h'<-1$). Since $h$ is a continuous and differentiable function, then by the mean value theorem, there exists $c \in [\phi, \phi + \tau]$ such that
$$
1 < \frac{h(1-\phi-\tau) - h(1-\phi)}{\tau} = -h'(1-c) \leq \lambda_0.
$$
Hence, the phase difference $\tau$ is scaled by at most a factor of $\lambda_0 >0$. Since $R = h(h(.))$ is an increasing function, the phase $\tau'$ at the end of \textit{full cycle} is at most $h\bigl(h(\lambda_0 \tau)\bigr) = R(\lambda_0 \tau)$. Similarly, if $t$ is after both users have sent a pulse, then $\tau'$ is at most $\lambda_0 h\bigl(h(\tau)\bigr) = \lambda_0 R(\tau)$. And if $t$ is after one of the users has sent a pulse, then $\tau'$ is actually reduced and lower bounded by $h\bigl(\lambda_0 h(\tau)\bigr)$. This completes the proof.
\end{proof}

Let $R_{\lambda}(\tau)\,\deff\, h\bigl(\lambda h(\tau)\bigr)$. Then we have the following lemma.  

\begin{lemma}
\label{last-lemma}
There exists at most one fixed point for each of the functions $R_{\lambda}(\tau)$, $\lambda R(\tau)$, and $R(\lambda \tau)$, for any $\lambda > 1$. Furthermore, if $R(\lambda \delta) < \delta$, then $R_{\lambda}(\tau)$, $\lambda R(\tau)$, and $R(\lambda \tau)$ have exactly one fixed point. 
\end{lemma}

\begin{proof}
Since $R'>1$ and $h'<-1$ over their domains, the derivatives of $R_{\lambda}(\tau)$, $\lambda R(\tau)$, and $R(\lambda \tau)$ are also greater than $1$ over their domains. Hence, the first part follows. Now, let $F(\tau) = R(\lambda \tau) - \tau$. Note that the domain of $R$ is $(\delta,h^{-1}(\delta))$ and it can be observed that $F(h^{-1}(\delta)) > 0$. Now, if $R(\lambda \delta) < \delta$ (equivalent to $F(\delta) < 0$) and since $F' > 0$, then there exists exactly one root for $F$ which becomes a fixed point for $\lambda R(\tau)$. Since $F$ is an increasing function, then $F(\delta/\lambda) < F(\delta) < 0$ or equivalently $\lambda R(\delta) < \delta$. Therefore, $\lambda R(\delta)$ has a unique fixed point using the same argument. Also, note that $R_{\lambda}(\delta) < \delta$ and it can be observed that $R_{\lambda}(h^{-1}(\delta)) > h^{-1}(\delta)$ is equivalent to $R(\lambda \delta) < \delta$ since $h$ is a decreasing function. Hence, $R_{\lambda}(\delta)$ has a unique fixed point using the same argument. 
\end{proof}

\begin{corollary}
\label{cor2}
Suppose that $R(\lambda \delta) < \delta$ and let $\tau_{\lambda}^*$ denote the fixed point for $R({\lambda}\tau)$. Then $\lambda \tau_{\lambda}^*$ is the fixed point for $\lambda R(\tau)$ and $h^{-1}(\tau_{\lambda}^*)$ is the fixed point for $R_{\lambda}(\tau)$. 
\end{corollary}

The following proposition is the main result of this section which shows that synchronization always occurs, under certain conditions, in the presence of jamming attacks. 

\begin{proposition}
\label{main-prop}
If $R(\lambda_0 \delta) < \delta$, where $\lambda_0$ is defined in \eq{lamb-def}, then there exists a $\tau^* \in (\delta,h^{-1}(\delta))$ such that for any initial phase difference $\tau$ with $\tau \notin (\tau^*, h^{-1}(\tau^*))$, synchronization always occurs in the presence of the jamming attack. 
\end{proposition}

\begin{proof}
Let $\tau^*$ denote the fixed point of $R(\lambda_0 \tau)$, which exists by \Lref{last-lemma}. Then for any $\lambda$, where $1 < \lambda < \lambda_0$, the fixed points of $R_{\lambda}(\tau)$, $\lambda R(\tau)$, and $R(\lambda \tau)$ belong to $(\tau^*, h^{-1}(\tau^*))$. Furthermore, since the derivatives of $R_{\lambda}(\tau)$, $\lambda R(\tau)$, and $R(\lambda \tau)$ are also greater than $1$, they are repeller functions. In particular, for $\tau > h^{-1}(\tau^*)$, we have $R_{\lambda}(\tau) > \tau$, and for $\tau < \tau^*$, we have $\max\{R(\lambda_0 \tau),\lambda_0 R(\tau)\} < \tau$. The proposition follows by this together with \Lref{lem-jam}. 
\end{proof}

The result of Proposition\,\ref{main-prop} can be also interpreted as follows. Assuming that $R(\lambda_0 \delta) < \delta$ holds, then synchronization occurs with probability at least $1-h^{-1}(\tau^*) + \tau^*$, where $\tau^*$ is the fixed point of $R(\lambda_0 \tau)$. Then the results of Section\,\ref{sec4} and Section\,\ref{sec5} can be extended to cases where a jammer is also present. In fact, we expect that these results will be scaled by a constant factor as the probability of synchronization and perhaps a modification of the parameters of the distribution of shared randomness is also needed.

\section{Conclusion}
\label{sec7}

In this paper, motivated by practical limitations of secret key generation protocols, we proposed to exploit readily available synchronization mechanisms in wireless networks, in particular pulse-coupled synchronization, for sharing common randomness between legitimate parties. The initial random phases of the POs deployed by the parties serve as the source of common randomness. Bounds on the entropy of such randomness, which is almost identically observed by the users, are derived for a two-user system. Furthermore, a three-terminal scenario is considered including two legitimate parties and a passive Eve. Eve's receiver is modeled and then a bound on the secret key rate is derived. Also, it is shown that, under certain conditions, the proposed protocol is resilient to active jamming with similar pulse generation mechanism.

There are several directions for future work. It is interesting to generalize the result of Section\,\ref{sec4} to networks with more than two users and, in particular, to check the validity the conjecture discussed in Remark\,1. In a more abstract setup and assuming a central user, this relates to the problem of distributed secret sharing in multi-user scenarios \cite{soleymani2018distributed}. The eavesdropper's model, described in Section\,\ref{sec3} and investigated in Section\,\ref{sec5}, can be extended to take into account memory, imperfectness of pulses, synchronization error, etc. The model for the jamming attack, discussed in Section\,\ref{sec6}, can be also extended to consider more general attacks such as a fixed-power interference, sending higher frequency pulses, etc. Furthermore, implementing the proposed system in front-end antennas, e.g., using setups similar to \cite{IMS}, is another interesting future direction. 

\bibliographystyle{IEEEtran}
\bibliography{PCO}

\begin{thebibliography}{10}
\providecommand{\url}[1]{#1}
\csname url@samestyle\endcsname
\providecommand{\newblock}{\relax}
\providecommand{\bibinfo}[2]{#2}
\providecommand{\BIBentrySTDinterwordspacing}{\spaceskip=0pt\relax}
\providecommand{\BIBentryALTinterwordstretchfactor}{4}
\providecommand{\BIBentryALTinterwordspacing}{\spaceskip=\fontdimen2\font plus
\BIBentryALTinterwordstretchfactor\fontdimen3\font minus
  \fontdimen4\font\relax}
\providecommand{\BIBforeignlanguage}[2]{{%
\expandafter\ifx\csname l@#1\endcsname\relax
\typeout{** WARNING: IEEEtran.bst: No hyphenation pattern has been}%
\typeout{** loaded for the language `#1'. Using the pattern for}%
\typeout{** the default language instead.}%
\else
\language=\csname l@#1\endcsname
\fi
#2}}
\providecommand{\BIBdecl}{\relax}
\BIBdecl

\bibitem{bloch2011physical}
M.~Bloch and J.~Barros, \emph{Physical-layer security: from information theory
  to security engineering}.\hskip 1em plus 0.5em minus 0.4em\relax Cambridge
  University Press, 2011.

\bibitem{ahlswede1993common}
R.~Ahlswede and I.~Csisz{\'a}r, ``Common randomness in information theory and
  cryptography. i. secret sharing,'' \emph{IEEE Transactions on Information
  Theory}, vol.~39, no.~4, pp. 1121--1132, 1993.

\bibitem{maurer1993secret}
U.~M. Maurer, ``Secret key agreement by public discussion from common
  information,'' \emph{IEEE transactions on information theory}, vol.~39,
  no.~3, pp. 733--742, 1993.

\bibitem{zhang2016key}
J.~Zhang, T.~Q. Duong, A.~Marshall, and R.~Woods, ``Key generation from
  wireless channels: A review,'' \emph{IEEE Access}, vol.~4, pp. 614--626,
  2016.

\bibitem{cgc}
H.~Liu, Y.~Wang, J.~Yang, and Y.~Chen, ``Fast and practical secret key
  extraction by exploiting channel response,'' in \emph{INFOCOM, 2013
  Proceedings IEEE}.\hskip 1em plus 0.5em minus 0.4em\relax IEEE, 2013, pp.
  3048--3056.

\bibitem{gollakota2011physical}
S.~Gollakota and D.~Katabi, ``Physical layer wireless security made fast and
  channel independent,'' in \emph{INFOCOM, 2011 Proceedings IEEE}.\hskip 1em
  plus 0.5em minus 0.4em\relax IEEE, 2011, pp. 1125--1133.

\bibitem{simeone2008distributed}
O.~Simeone, U.~Spagnolini, Y.~Bar-Ness, and S.~H. Strogatz, ``Distributed
  synchronization in wireless networks,'' \emph{IEEE Signal Processing
  Magazine}, vol.~25, no.~5, 2008.

\bibitem{glob2}
N.~Aldaghri and H.~Mahdavifar, ``Fast secret key generation in static
  environments using induced randomness,'' \emph{Proceedings of IEEE Globcom},
  Abu Dhabi, UAE, Dec 2018.

\bibitem{glob1}
N.~Ebrahimi, H.~Mahdavifar, and E.~Afshari, ``A novel approach to secure
  communication in physical layer via coupled dynamical systems,''
  \emph{Proceedings of IEEE Globcom}, Abu Dhabi, UAE, Dec 2018.

\bibitem{ebrahimi2016robustness}
N.~Ebrahimi and J.~Buckwalter, ``Robustness of injection-locked oscillators to
  {CMOS} process tolerances,'' in \emph{International Conference on
  Applications in Nonlinear Dynamics}.\hskip 1em plus 0.5em minus 0.4em\relax
  Springer, 2016, pp. 245--263.

\bibitem{ebrahimi201771}
N.~Ebrahimi, P.-Y. Wu, M.~Bagheri, and J.~F. Buckwalter, ``A 71--86 {GHz}
  phased array transceiver using wideband injection-locked oscillator phase
  shifters,'' \emph{IEEE Transactions on Microwave Theory and Techniques},
  vol.~65, no.~2, pp. 346--361, 2017.

\bibitem{str}
R.~E. Mirollo and S.~H. Strogatz, ``Synchronization of pulse-coupled biological
  oscillators,'' \emph{SIAM Journal on Applied Mathematics}, vol.~50, no.~6,
  pp. 1645--1662, 1990.

\bibitem{peskin1975mathematical}
C.~S. Peskin, ``Mathematical aspects of heart physiology,'' \emph{Courant Inst.
  Math}, pp. 268--278, 1975.

\bibitem{csiszar2004secrecy}
I.~Csisz{\'a}r and P.~Narayan, ``Secrecy capacities for multiple terminals,''
  \emph{IEEE Transactions on Information Theory}, vol.~50, no.~12, pp.
  3047--3061, 2004.

\bibitem{ernst1995synchronization}
U.~Ernst, K.~Pawelzik, and T.~Geisel, ``Synchronization induced by temporal
  delays in pulse-coupled oscillators,'' \emph{Physical review letters},
  vol.~74, no.~9, p. 1570, 1995.

\bibitem{ferrari2017convergence}
L.~Ferrari, A.~Scaglione, R.~Gentz, and Y.-W.~P. Hong, ``Convergence results on
  pulse coupled oscillator protocols in locally connected networks,''
  \emph{IEEE/ACM Trans. on Networking}, vol.~25, no.~2, pp. 1004--1019, 2017.

\bibitem{chou2015polar}
R.~A. Chou, M.~R. Bloch, and E.~Abbe, ``Polar coding for secret-key
  generation,'' \emph{IEEE Transactions on Information Theory}, vol.~61,
  no.~11, pp. 6213--6237, 2015.

\bibitem{arikan2009channel}
E.~Arikan, ``Channel polarization: A method for constructing capacity-achieving
  codes for symmetric binary-input memoryless channels,'' \emph{IEEE
  Transactions on Information Theory}, vol.~55, no.~7, pp. 3051--3073, 2009.

\bibitem{mahdavifar2011achieving}
H.~Mahdavifar and A.~Vardy, ``Achieving the secrecy capacity of wiretap
  channels using polar codes,'' \emph{IEEE Transactions on Information Theory},
  vol.~57, no.~10, pp. 6428--6443, 2011.

\bibitem{soleymani2018distributed}
M.~Soleymani and H.~Mahdavifar, ``Distributed multi-user secret sharing,''
  \emph{arXiv preprint arXiv:1801.04384}, 2018.

\bibitem{IMS}
N.~Ebrahimi, B.~Yektakhah, K.~Sarabandi, H.-S. Kim, D.~D. Wentzloff, and
  D.~Blaauw, ``A novel physical layer security technique using master-slave
  full duplex communication,'' \emph{Proceedings of IEEE/MTT-S International
  Microwave Symposium (IMS)}, Boston, MA, June 2019.

\end{thebibliography}

\end{document}